\documentclass[a4paper,fleqn]{cas-dc}

\usepackage{flushend}
\graphicspath{{figures/}}
\usepackage[numbers,sort&compress]{natbib}

\newtheorem{ass}{Assumption}
\newtheorem{defn}{Definition}
\newtheorem{thm}{Theorem}
\newtheorem{prop}{Proposition}
\newtheorem{lem}{Lemma}

\newtheorem{rem}{Remark}
\newtheorem{claim}{Claim}
\newtheorem{cexample}{Counterexample}
\newproof{pf}{Proof}
\DeclareMathOperator{\artanh}{artanh}

\begin{document}
\let\WriteBookmarks\relax
\def\floatpagepagefraction{1}
\def\textpagefraction{.001}

% Short title and authors
\shorttitle{Turnpike properties in nonlinear system identification}    
\shortauthors{J. D. Schiller \& M. A. Müller}  

% Main title and footnote
\title [mode = title]{Turnpike properties in nonlinear system identification}  
\tnotemark[1] 
\tnotetext[1]{
	This work was supported by the Deutsche Forschungsgemeinschaft (DFG, German Research Foundation), Project 535860958.
} 

% First author
\author[]{Julian D. Schiller}[orcid=0000-0002-7091-0849]
\cormark[1]
\cortext[1]{Corresponding author}
\ead{schiller@irt.uni-hannover.de}
\credit{Conceptualization, Formal analysis, Methodology, Project administration, Software, Validation, Visualization, Writing -- original draft}

% Second author
\author[]{Matthias A. Müller}[orcid=0000-0002-4911-9526]
\ead{mueller@irt.uni-hannover.de}
\credit{Conceptualization, Funding acquisition, Project administration, Resources, Validation, Writing -- Review \& Editing}

% Address/affiliation
\affiliation[]{
	organization={Leibniz University Hannover, Institute of Automatic Control},
	city={Hannover}, country={Germany}
}

\begin{abstract}
	We analyze the problem of learning general discrete-time nonlinear state-space models using the simulation error minimization (SEM) method. In this setting, model parameters are typically learned by minimizing the mismatch between simulated and measured outputs over a training dataset, or shorter subsequences extracted from it. Specifically, we study the cumulative output turnpike property of the underlying SEM optimization problem, which requires optimal output sequences emanating from a fixed initial state to approach and remain close to an optimal output sequence of the corresponding SEM problem with free initial state. In the presence of non-unique optimal output sequences---as may arise, for instance, in fully black-box system identification using neural networks---the property is formulated with respect to the closest such sequence. Turnpike behavior is generally desirable in practice, as it provides a theoretical justification for employing computationally more tractable SEM formulations with fixed initial states while ensuring that their optimal output sequences remain close to unconstrained optimal ones. Under a mild reachability condition, we establish equivalence between the cumulative turnpike property, coercivity of the value function, and a tailored notion of strict dissipativity. We additionally introduce a cardinality turnpike property and show that it is strictly weaker than the cumulative notion. Finally, we establish sufficient conditions for turnpike behavior based on incremental output stability, convexity of the stage cost, and a suitable optimality condition, and illustrate the theory by means of a numerical example.
\end{abstract}

% Keywords

\maketitle

%%%%%%%%%%%%%%%%%%%%%%%%%%%%%%%%%%%%%%%%%%%%%%%%%%%%%%%%%%%%%%%%%%%%%%%%%%%%%%%%%%%%

\section{Introduction}

Learning state-space models from input-output data is a fundamental problem in system identification. A classical approach is prediction error minimization, where one-step-ahead predictors are identified by minimizing the corresponding prediction error \cite{Ljung1999}.
In many applications, however, accurate long-term predictions over a finite horizon are required, e.g., for model analysis, simulation, model-based control design, and model predictive control. In such settings, simulation error minimization (SEM) is often more appropriate, as it directly minimizes multi-step prediction errors \cite{Farina2010c,Aguirre2010}.

While conceptually appealing, solving the full SEM problem is computationally challenging, particularly in the presence of highly nonlinear models and large datasets. Specifically, when employing standard iterative gradient-based numerical optimization methods, each iteration requires performing the full forward simulation of the model together with the associated sensitivity (backward) computations over the entire dataset, and may hence lead to numerical difficulties and slow training \cite{Ribeiro2020,AllenZhu2019,Ayed2019}. A common practical alternative is therefore to perform SEM over shorter subsequences extracted from the dataset, enabling efficient and parallelizable training on modern hardware.

A central issue in SEM is the treatment of the initial state. While it can be included as a decision variable and learned jointly with the model parameters, this significantly increases computational complexity and hinders parallelization. Alternative approaches include online estimation techniques such as extended Kalman filtering \cite{Bemporad2023}, or regularized formulations that retain parallelizability while ensuring consistency \cite{Forgione2021}, albeit at increased computational cost.
A widely used practical approach is to simply fix the initial state---e.g., randomly or to zero---and optimize only over the model parameters, as is common in training recurrent neural networks; see, e.g., \cite{Bonassi2022,Jaeger2002,Schiller2025b,Mohajerin2019,Forgione2022}.

More sophisticated approaches augment this setup with an encoder that maps past input-output data to an internal state used to initialize the model; see, e.g., \cite{Masti2021,Beintema2023,Beintema2023a,Bemporad2025}. Since the encoder is usually trained jointly with the model, this increases the computational burden and may introduce additional hyperparameters in the composite loss function that require careful tuning; see, e.g., \cite{Masti2021,Mohajerin2019}.

These computationally motivated approaches naturally raise the question of how fixing or estimating the initial state affects the resulting SEM solution. In this context, turnpike theory provides a useful framework for analysis. Originating in optimal control, the turnpike property describes the phenomenon that optimal trajectories spend most of their time close to a particular steady state or reference trajectory, referred to as the \emph{turnpike}  \cite{Faulwasser2022,Gruene2016a,Gruene2018,Zaslavski2006,Trelat2018a}.
This concept can be adapted to system identification by formulating the turnpike property in terms of outputs rather than states and controls, thereby accounting for the inherent non-uniqueness of model parameterizations~\cite{Schiller2025b}.
The resulting framework enables the analysis of discrepancies between output sequences generated by SEM solutions with an arbitrarily fixed initial state and those associated with the desired unconstrained SEM problem with free initial state.

\paragraph{Contributions.}
In this paper, we generalize the concept of cumulative output turnpike to the case of non-unique optimal output sequences.
This setting naturally captures training of modern deep learning models using SEM methods, where distinct optimal model parameterizations may generate different output sequences while attaining minimal cost.
The presence of turnpike behavior is generally desirable in practice, as it allows one to employ tractable (potentially truncated) SEM methods with a fixed initial state while guaranteeing that the resulting output sequence remains close to an unconstrained optimal one.
Under a mild reachability condition, we show in Section~\ref{sec:main_result} that the cumulative turnpike property is equivalent to a coercivity property of the value function and a tailored notion of strict dissipativity.
Moreover, in Section~\ref{sec:cardTP}, we introduce a corresponding cardinality turnpike property and show that cumulative turnpike implies cardinality turnpike, whereas the converse does not hold in general---equivalence can, however, be recovered under suitable boundedness assumptions. In Section~\ref{sec:sysid_practice}, we establish sufficient conditions for turnpike behavior in system identification based on incremental output stability, convexity of the stage cost, and a suitable optimality condition. Finally, in Section~\ref{sec:example}, we illustrate the theory using a simple Elman-type recurrent neural network for which these conditions can be verified analytically.

\paragraph{Notation.}
The set of non-negative natural numbers (including zero) is denoted by $\mathbb{N}$, and the set of positive natural numbers by $\mathbb{N}_{+}:=\mathbb{N}\setminus\{0\}$.
The interval notation $[a,b]$ with $a,b\in\mathbb{N}$ refers to integers. 
The Euclidean norm of a vector $x\in\mathbb{R}^n$ is denoted by $\|x\|$.
For a finite set \(\mathcal A\), its cardinality is denoted by \(\#\mathcal A\).
A continuous function $\alpha : \mathbb{R}_{\geq0} \to \mathbb{R}_{\geq0}$ is said to belong to class $\mathcal{K}$ if it is strictly increasing and satisfies $\alpha(0)=0$. It is of class $\mathcal{K}_\infty$ if, in addition, $\alpha(s) \to \infty$ as $s \to \infty$.
A continuous function $\sigma : \mathbb{R}_{\geq0} \to \mathbb{R}_{\geq0}$ is said to belong to class $\mathcal{L}$ if it is non-increasing and satisfies $\lim_{s \to \infty} \sigma(s) = 0$.

\section{Problem setting}

In this section, we introduce the considered SEM problem, state the standing assumptions, and establish the basic properties required for the subsequent analysis.

\subsection{Preliminaries}
In this work, we consider the problem of learning the parameters of a general nonlinear discrete-time state-space model in the form of
\begin{align}
	x^+ &= f(x,u;\theta),\label{eq:sys_1}\\
	y &= h(x,u;\theta),\label{eq:sys_2}
\end{align}
where $x\in\mathbb{R}^n$ is the state of the system, $x^+\in\mathbb{R}^n$ is the successor state, $u\in\mathbb{R}^m$ is the input, and $y\in\mathbb{R}^p$ is the output. The dynamics and output equation of the model are characterized by nonlinear continuous functions $f$ and $h$, which are parameterized by a parameter vector $\theta\in\mathbb{R}^o$.
In the following, we will frequently use $y_j(\bar{x},\theta,U)$ to denote the output generated by the model \eqref{eq:sys_1}--\eqref{eq:sys_2} at time $j\in\mathbb{N}$ under a particular parameterization $\theta$, when being initialized with $x_0 = \bar{x}$ and driven by the input sequence $U=\{u_j\}_{j=0}^\infty$.

To determine the model parameters $\theta$, we assume that there is an input-output training data sequence $D$ of length $N+1$ with $N\in\mathbb{N}$:
\begin{equation}\label{eq:Di}
	D := \{({u}_{j}^\mathrm{d},{y}_{j}^\mathrm{d})\}_{j=0}^N, \quad U^\mathrm{d} := \{u^\mathrm{d}_{j}\}_{j=0}^N.
\end{equation}
In line with standard learning techniques, we suppose that the training data is normalized in the sense that, for any $N\in\mathbb{N}$, $({u}_{j}^\mathrm{d},{y}_{j}^\mathrm{d})\in\mathcal{U}^\mathrm{d}\times\mathcal{Y}^\mathrm{d}$ for all $j\in[0,N]$, where $\mathcal{U}^\mathrm{d}\subset\mathbb{R}^{m}$ and $\mathcal{Y}^\mathrm{d}\subset\mathbb{R}^{p}$ are compact.

\subsection{System identification}
To learn the parameters of the model~\eqref{eq:sys_1}--\eqref{eq:sys_2} using SEM, we minimize the following cost function:
\begin{equation}\label{eq:NLP_cost}
	J(x_0,\theta) := \sum_{j=0}^N \beta_j l({y}_j(x_0,\theta,U^\mathrm{d}),{y}_{j}^\mathrm{d}) + r(\theta,\bar{\theta}).
\end{equation}
Here, $l:\mathbb{R}^p\times\mathbb{R}^p\rightarrow\mathbb{R}_{\geq0}$ is a continuous stage cost, measuring the mismatch between model outputs and the data. The function $r:\mathbb{R}^o\times\mathbb{R}^o\rightarrow \mathbb{R}_{\geq0}$ is a regularization cost (with $\bar{\theta}$ being a fixed point in the parameter space, usually taken as the origin) that can be incorporated to incentivize additional model properties (such as sparsity, low complexity, or parameter smallness), compare also \cite[Sec.~4, p.~221]{Ljung1999} and \cite{Pillonetto2025,Bemporad2023,Bemporad2025}. Note that we allow for arbitrary lower semicontinuous functions $r$ in~\eqref{eq:NLP_cost}, which covers common choices such as $L_2$, $L_1$, and $L_0$ regularization, or $r\equiv0$ (i.e., without any parameter regularization).

The criterion in \eqref{eq:NLP_cost} also includes a scalar weighting $\beta:[0,N]\rightarrow \mathbb{R}_{\geq0}$. This can be leveraged if measurements are considered to be of varying reliability or relevance \cite[Sec.~7.2]{Ljung1999}, or to include a burn-in phase to reduce the influence of model transients; compare \cite{Schiller2025b,Jaeger2002,Bonassi2022}.
Note that $\beta_j$ may also depend on $N$, which allows for the use of normalized cost functions; compare also Remarks~\ref{rem:beta} and~\ref{rem:burn_in} below.

Given a dataset $D$ as in~\eqref{eq:Di}, the model parameters~$\theta$ are then obtained by solving the following SEM optimization problem:%
\begin{subequations}
	\label{eq:NLP}
	\begin{align}
		&\min\limits_{x_0,\theta} J(x_0,\theta) \label{eq:NLP_min}\\
		\text{s.t. \quad }
		& x_{j+1} = f(x_j,u_j^\mathrm{d};\theta), \ j\in[0,N], \label{eq:NLP_f}\\
		& y_{j} = h(x_{j},u_j^\mathrm{d};\theta), \ j\in[0,N],\label{eq:NLP_h}\\
		& x_{j+1}\in\mathcal{X}, \ y_{j} \in \mathcal{Y},\ j\in[0,N] \label{eq:NLP_X}\\
		& x_0\in\mathcal{X}_0, \ \theta \in \Theta. \label{eq:NLP_Theta}
	\end{align}
\end{subequations}
Here, the constraints~\eqref{eq:NLP_f}--\eqref{eq:NLP_h} ensure consistency with the model equations~\eqref{eq:sys_1}--\eqref{eq:sys_2}.
We optimize $(x_0,\theta)$ over compact sets $\mathcal{X}_0\subset\mathbb{R}^n$ and $\Theta\subset\mathbb{R}^o$, implemented via~\eqref{eq:NLP_Theta}.
The constraints in~\eqref{eq:NLP_X} can be used to incorporate additional (e.g., physical) knowledge by specifying closed sets $\mathcal{X}\subset\mathbb{R}^n$ and/or $\mathcal{Y}\subset\mathbb{R}^p$; otherwise, these may be omitted (i.e., set to $\mathcal{X}=\mathbb{R}^n$ and $\mathcal{Y}=\mathbb{R}^p$).
We make the standing assumption that for each fixed initial condition $x_0\in \mathcal{X}_0$, there exists $\theta\in\Theta$ such that the problem~\eqref{eq:NLP} is feasible.

Due to the fact that the cost function in~\eqref{eq:NLP_cost} is lower semicontinuous, the sets $\mathcal{X}_0$ and $\Theta$ are compact by assumption, and the problem is feasible, by \cite[Th.~1.9]{Rockafellar1998} there exists an optimal solution $(x_0^*,\theta^*)$ that achieves
\begin{equation}\label{eq:NLP_case_TP}
	J(x_0^*,\theta^*) = V^* := \min_{x_0,\theta}\{ J(x_0,\theta)\ |\  \text{\eqref{eq:NLP_f}--\eqref{eq:NLP_Theta}}\}<\infty.
\end{equation}
However, the solution $(x_0^*,\theta^*)$ may be non-unique and therefore belongs to a corresponding \emph{solution set}
\begin{equation}\label{eq:S_star}
	\mathcal{S}^*:=\{(x_0,\theta)\in\mathcal{X}_0\times\Theta \ | \ J(x_0,\theta) = V^*, \text{\eqref{eq:NLP_f}--\eqref{eq:NLP_Theta}}\},
\end{equation}
including all optimal pairs of initial conditions and parameters that achieve the optimal cost $V^*$ while satisfying the constraints \eqref{eq:NLP_f}--\eqref{eq:NLP_Theta}. The conditions above ensure that $\mathcal{S}^*$ is non-empty and compact.

The SEM problem in~\eqref{eq:NLP_case_TP} involves joint optimization over parameters and initial conditions. While theoretically appealing, this drastically increases the complexity of the optimization problem. A simple yet effective approach is to fix the initial state (e.g., to zero, a random value, or a measured or estimated state) and optimize only over the parameters, see, e.g., \cite{Xu2023a,AllenZhu2019,Schiller2025b,Jaeger2002,Bonassi2022,Mohajerin2019,Masti2021,Beintema2023a,Beintema2023,Forgione2022}.
{Motivated by this, we consider a constrained SEM problem, where we fix the initial condition in~\eqref{eq:NLP} to some pre-defined $\bar{x}\in\mathcal{X}_0$.}
Following similar arguments as outlined above~\eqref{eq:NLP_case_TP}, for each $\bar{x}\in\mathcal{X}_0$ there exists an optimal parameter $\theta'$ such that
\begin{equation}\label{eq:NLP_case_1}
	J(\bar{x},\theta') = V'(\bar{x}) := \min_{\theta}\{ J(\bar{x},\theta)\ |\  \text{\eqref{eq:NLP_f}--\eqref{eq:NLP_Theta}}\} < \infty.
\end{equation}
Again, the solution $\theta'$ may be non-unique and therefore belongs to the corresponding solution set
\begin{equation}\label{eq:S_prime}
	\mathcal{S}'(\bar{x}):=\{\theta\in\Theta \ | \ J(\bar{x},\theta) = V'(\bar{x}), \text{\eqref{eq:NLP_f}--\eqref{eq:NLP_Theta}}\}.
\end{equation}

\begin{rem}[Truncated SEM]
	Our framework and resulting analysis are directly applicable to truncated SEM. Here, one seeks to minimize the simulation error over shorter, possibly overlapping subsequences extracted from the training dataset $D$, thereby improving numerical stability, efficiency, and training speed, see, e.g.,~\cite{Bonassi2022,Jaeger2002,Schiller2025b,Mohajerin2019,Forgione2022}.
	In this setting, all variables and functions in~\eqref{eq:sys_1}--\eqref{eq:NLP} are interpreted as a lifted (stacked) form over the collection of subsequences, coupled through a common parameter~$\theta$. This hence enables a systematic analysis of the effects induced by truncation; see also~\cite{Schiller2025b} for more details.
\end{rem}

\section{The turnpike phenomenon in system identification}\label{sec:main_results}

In this section, we investigate conditions under which a solution $\theta'\in\mathcal{S}'(\bar{x})$ generates an output sequence that remains close to an optimal output sequence associated with $\mathcal{S}^*$.
We characterize this notion of closeness using the cumulative turnpike property in Section~\ref{sec:turnpike}, introduce related system-theoretic properties in Section~\ref{sec:properties}, establish their equivalence in Section~\ref{sec:main_result}, and relate the cumulative notion to the weaker cardinality turnpike property in Section~\ref{sec:cardTP}.

\subsection{Cumulative turnpike characterization}\label{sec:turnpike}

We first introduce the following output-based distance measure between two pairs $(x_0,\theta),(\hat{x}_0,\hat{\theta})\in\mathcal{X}_0\times\Theta$:
\begin{equation}\label{eq:J_tilde}
	D(({x}_0,{\theta}),(\hat{x}_0,\hat{\theta})):= \sum_{j=0}^N \beta_j\|y_j({x}_0,{\theta},U^\mathrm{d})-y_j(\hat{x}_0,\hat{\theta},U^\mathrm{d})\|.
\end{equation}
For a given pair $(x_0,\theta)\in\mathcal{X}_0\times\Theta$, we are interested in its distance to the optimal solution set $S^*$ and, in particular, in those optimal pairs attaining the minimum distance. These generally form a subset $\mathcal{S}^*(x_0,\theta)\subseteq \mathcal{S}^*$ that can be defined as
\begin{equation}
	\mathcal{S}^*(x_0,\theta) := \arg\min_{(\hat{x}_0,\hat{\theta})\in\mathcal{S}^*}	D(({x}_0,{\theta}),(\hat{x}_0,\hat{\theta})). \label{eq:set_S0x}
\end{equation}
Since $S^*$ is compact and $D$ is continuous, $S^*(x_0,\theta)$ is non-empty.

We now introduce the cumulative turnpike property, which quantifies the discrepancy between the outputs generated by a constrained SEM solution $\theta'\in\mathcal{S}'(\bar{x})$ and the outputs generated by a corresponding closest optimal pair $(x_0^*,\theta^*)\in\mathcal{S}^*(\bar{x},\theta')$.

\begin{defn}[Cumulative turnpike]\label{def:intTP}
	The problem~\eqref{eq:NLP_case_1} has the \emph{cumulative turnpike property} if there exists $\alpha\in\mathcal{K}_\infty$, and $E>0$ such that for all $\bar{x}\in\mathcal{X}_0$ and all $N\in\mathbb{N}$,
	\begin{equation}\label{eq:intTP}
		\sum_{j=0}^{N}\alpha(\|y_j(\bar{x},\theta',U^\mathrm{d})-y_j(x_0^*,\theta^*,U^\mathrm{d})\|) \leq E
	\end{equation}
	for all $\theta'\in\mathcal{S}'(\bar{x})$ and all $(x_0^*,\theta^*)\in \mathcal{S}^*(\bar{x},\theta')$.
\end{defn}

Definition~\ref{def:intTP} requires the cumulative discrepancy between the outputs of the constrained and unconstrained SEM solutions to remain uniformly bounded with respect to the dataset size $N$.
In particular,~\eqref{eq:intTP} implies that
\begin{equation*}
	\lim_{N \to \infty}
	\frac{1}{N+1}\sum_{j=0}^{N}
	\alpha\left(
	\|y_j(\bar{x},\theta',U^{\mathrm d})
	-y_j(x_0^*,\theta^*,U^{\mathrm d})\|
	\right)
	=0,
\end{equation*}
that is, the average output discrepancy measured through $\alpha$ vanishes as $N\to\infty$. Intuitively, the constrained optimal output sequence therefore remains close to an optimal output sequence of the unconstrained SEM problem for most time indices, while their difference remains uniformly bounded at every time index. Here, the optimal output sequence of the unconstrained SEM problem is referred to as the \emph{turnpike}. In the non-unique case, this property is understood with respect to a closest optimal output sequence in terms of $D$.

\begin{rem}[Output turnpike]
	The turnpike property in Definition~\ref{def:intTP} is formulated in terms of outputs, rather than optimal states, controls, or adjoints, as commonly done in optimal control; see, e.g., \cite{Trelat2025,Faulwasser2022,Zaslavski2006}.
	This is natural in system identification, where the set of initial conditions and parameters generating the same output sequence is generally \emph{not} a singleton, unless sufficiently strong identifiability conditions are imposed.
\end{rem}

The presence of turnpike behavior is generally desirable in practice, as it allows one to employ the constrained SEM problem---which is computationally more tractable---while ensuring that the resulting output sequence stays close to an optimal one.
In the following, we establish connections to related system-theoretic properties such as coercivity of the value function and strict dissipativity. In Section~\ref{sec:sysid_practice}, we further show how these properties (and hence turnpike behavior) arise under suitable assumptions.

\subsection{Related system-theoretic properties}\label{sec:properties}

In the following, we introduce two related system-theoretic properties of the optimization problem~\eqref{eq:NLP}. To account for different scalings of the SEM cost, we introduce the average weight
\begin{equation}\label{eq:B_def}
	M(N):=\frac{1}{N+1}\sum_{j=0}^N\beta_j.
\end{equation}
Throughout the following, we assume that $M(N)>0$ for all $N\in\mathbb{N}$.
The first property corresponds to coercivity of the value function in~\eqref{eq:NLP_case_1}, involving $V^*$ and the set of output turnpikes.

\begin{defn}[Coercivity]\label{def:coerc}
	The value function in~\eqref{eq:NLP_case_1} is \emph{coercive} if there exist $\gamma,\alpha_\mathrm{c}\in\mathcal{K}_\infty$ and $C_\mathrm{c}\geq0$ such that for all $N\in\mathbb{N}$ and $\bar{x}\in\mathcal{X}_0$, it holds that
	\begin{equation}\label{eq:coerc}
		V'(\bar{x}) \geq V^* + M(N)\left(\gamma\left( \sum_{j=0}^{N} \alpha_\mathrm{c}(\|y'_j - y^*_j\|)\right) - C_\mathrm{c}\right)
	\end{equation}
	for all $\theta'\in\mathcal{S}'(\bar{x})$ and all $(x_0^*,\theta^*)\in \mathcal{S}^*(\bar{x},\theta')$, where $y_j' = y_j(\bar{x},\theta',U^\mathrm{d})$ and $y_j^* = y_j(x_0^*,\theta^*,U^\mathrm{d})$, $j\in[0,N]$.
\end{defn}

Definition~\ref{def:coerc} requires the value-function gap $V'(\bar{x})-V^*$, after normalization by $M(N)$, to grow with the cumulative output discrepancy between a constrained optimal solution and its corresponding turnpike, up to the offset $C_\mathrm{c}$. A similar property was previously introduced in \cite[Ass.~\textit{H}\textsubscript{4}]{Trelat2018a} in the context of continuous-time optimal control systems.

Turnpike behavior of optimal control problems is known to be closely related to the concept of dissipativity \cite{Trelat2025,Faulwasser2022,Gruene2018,Gruene2016a}. Inspired by this, we introduce a tailored notion of (strict) dissipativity of the SEM problem in~\eqref{eq:NLP_case_1}. To this end, for each $N\in\mathbb{N}$, we introduce the time-varying supply rate $s:[0,N]\times\mathcal{X}_0\times\Theta\times\mathcal{X}_0\times\Theta\rightarrow\mathbb{R}$ such that for all $j\in[0,N]$ and any two pairs $(x_0,\theta)\in\mathcal{X}_0\times\Theta$ and $(x_0^*,\theta^*)\in\mathcal{S}^*(x_0,\theta)$, it holds that
\begin{equation}
	s_j((x_0,\theta),(x_0^*,\theta^*)):= \beta_j \Delta l_j + \frac{1}{N+1}\Delta r \label{eq:diss_supply}
\end{equation}
with $\Delta l_j := l(y_j(x_0,\theta,U^\mathrm{d}),{y}_{j}^\mathrm{d}) - l(y_j(x_0^*,\theta^*,U^\mathrm{d}),{y}_{j}^\mathrm{d})$ and $\Delta r := r(\theta,\bar{\theta}) - r(\theta^*,\bar{\theta})$.
Note that by construction, for $(x_0,\theta)\in\mathcal{X}_0\times\Theta$ and $(x_0^*,\theta^*)\in\mathcal{S}^*(x_0,\theta)$, summing the supply rate $s_j$ over $j\in[0,N]$ yields the corresponding cost difference
\begin{equation}\label{eq:supply_value}
	\sum_{j=0}^{N}s_j = J(x_0,\theta)-V^*.
\end{equation}

\begin{defn}[Strict dissipativity]\label{def:diss}
	The problem in~\eqref{eq:NLP_case_1} is \emph{strictly dissipative} along solutions if there exist $\alpha_\lambda\in\mathcal{K}_\infty$ and $E_\lambda>0$ such that, for each $N\in\mathbb{N}$, there exists a time-varying storage function $\lambda{:}[0,N+1]\times\mathcal{X}_0\times\Theta\times\mathcal{X}_0\times\Theta{\,\rightarrow\,}\mathbb{R}$ such that for all $\bar{x}\in\mathcal{X}_0$, all $\theta'\in\mathcal{S}'(\bar{x})$, and all $(x_0^*,\theta^*)\in \mathcal{S}^*(\bar{x},\theta')$, it holds that $\lambda_0\leq M(N) E_\lambda$, $\lambda_{N+1}=0$, and
	\begin{align}
		&\lambda_{j+1}((\bar{x},\theta'),(x_0^*,\theta^*)) - \lambda_j((\bar{x},\theta'),(x_0^*,\theta^*)) \nonumber \\
		&\leq s_j((\bar{x},\theta'),(x_0^*,\theta^*))   \label{eq:diss_dissip}\\
		&\quad - M(N)\alpha_\lambda(\|y_j(\bar{x},\theta',U^\mathrm{d})- y_j(x_0^*,\theta^*,U^\mathrm{d})\|)  \nonumber
	\end{align}
	for all $j\in[0,N]$.
\end{defn}

We point out two main conceptual differences compared to dissipativity notions commonly used in turnpike and optimal control theory.
First, the storage function and supply rate are defined relative to a particular optimal pair $(x_0^*,\theta^*)$. Second, they may depend on the dataset size $N$.
Both aspects are natural in system identification: the turnpike may be non-unique and generally changes with the dataset size, since increasing the dataset may alter the optimal pair $(x^*_0,\theta^*)$ and hence the corresponding output sequence over the entire interval $[0,N]$. Importantly, the quantities $\alpha_\lambda$ and $E_\lambda$ in Definition~\ref{def:diss} are nevertheless required to be uniform with respect to $N$ and the considered optimal solutions.

\subsection{Equivalent characterizations}\label{sec:main_result}

In this section, we present our main result, which provides equivalent characterizations of the cumulative turnpike property. To this end, we impose a cost-reachability condition ensuring that the additional cost induced by fixing the initial state remains uniformly bounded relative to the cost scaling $M(N)$.

\begin{ass}[Cost reachability]\label{ass:reachability}
	There exists a constant $E_\mathrm{r}\geq0$ such that
	\begin{equation}\label{eq:reachability}
		V'(\bar{x}) - V^* \leq M(N)\cdot E_\mathrm{r}
	\end{equation}
	for all $\bar{x}\in\mathcal{X}_0$ and all $N\in\mathbb{N}$.
\end{ass}

We point out that similar reachability-type conditions, including related controllability or stabilizability properties, are standard in turnpike theory and are generally necessary for turnpike behavior in the context of optimal control; see, e.g., \cite[Th.~4.7]{Gruene2016a} and \cite[Ass.~1]{Gruene2018}.
In Section~\ref{sec:reachability}, we provide a sufficient condition for Assumption~\ref{ass:reachability} based on incremental output stability.

With this assumption in place, we can now state the following equivalence result.

\begin{thm}\label{thm:equivalences}
	Let Assumption~\ref{ass:reachability} be satisfied. Then, the following statements are equivalent:
	\begin{enumerate}[a)]
		\item Problem~\eqref{eq:NLP_case_1} has the cumulative turnpike property (Definition~\ref{def:intTP}).
		\item The value function in~\eqref{eq:NLP_case_1} is coercive (Definition~\ref{def:coerc}).
		\item Problem~\eqref{eq:NLP_case_1} is strictly dissipative (Definition~\ref{def:diss}).
	\end{enumerate}
\end{thm}

Theorem~\ref{thm:equivalences} provides equivalent characterizations of cumulative turnpike behavior in the considered non-unique setting in terms of properties of the SEM value function and the underlying optimization problem.
In particular, under cost reachability, the turnpike property can be established by verifying coercivity or, equivalently, strict dissipativity. The former characterization is exploited in Section~\ref{sec:sysid_practice}, where we provide sufficient conditions for cost reachability and coercivity that can be related to properties arising naturally in system identification. Theorem~\ref{thm:equivalences} is proven in Appendix~\ref{sec:thm_equivalences}.

\subsection{On the cardinality turnpike property}\label{sec:cardTP}

A significant part of the turnpike and optimal control literature considers characterizations that bound the number of time indices at which an optimal trajectory lies outside a neighborhood of the turnpike, commonly referred to as the \emph{cardinality turnpike property} (or, in continuous time, the \emph{measure turnpike property}); see, e.g., \cite{Faulwasser2022} and compare also \cite{Gruene2018,Zaslavski2006,Gruene2016a,Trelat2018a}.
In this section, we extend this notion to the context of system identification and relate it to the cumulative turnpike property from Definition~\ref{def:intTP}. In particular, we show that cumulative turnpike implies cardinality turnpike (Proposition~\ref{prop:intTP2cardTP}), whereas the converse does not hold in general (see Counterexample~\ref{counterex} and Remark~\ref{rem:cardTP_converse}).
Equivalence can, however, be recovered under suitable boundedness assumptions (Proposition~\ref{prop:cardTP2intTP}).

\begin{defn}[Cardinality turnpike]\label{def:cardTP}
	The problem~\eqref{eq:NLP_case_1} has the \emph{cardinality turnpike property} if there exists $\sigma\in\mathcal{L}$ such that for all $\bar{x}\in\mathcal{X}_0$, all $N\in\mathbb{N}$, and all $P\in\mathbb{N}_+$, it holds that%
	\begin{equation}\label{eq:TP_card}
		\#\mathcal{Q}(P,N)\leq P,
	\end{equation}
	for all $\theta'\in\mathcal{S}'(\bar{x})$ and all $(x_0^*,\theta^*)\in\mathcal{S}^*(\bar{x},\theta')$, where
	\begin{equation}\label{eq:TP_card_Q}
		\begin{split}
		\mathcal{Q}(P,N)&:=\big\{j\in[0,N] \mid \\
		& \|y_j(\bar{x},\theta',U^\mathrm{d})-y_j(x_0^*,\theta^*,U^\mathrm{d})\|>\sigma(P)\big\}.
		\end{split}
	\end{equation}
\end{defn}

Definition~\ref{def:cardTP} requires that, for any $P\in\mathbb{N}_+$, the constrained optimal output $y_j(\bar{x},\theta',U^\mathrm{d})$ lies outside a tube of radius $\sigma(P)$ around the corresponding turnpike $y_j({x}_0^*,\theta^*,U^\mathrm{d})$ at no more than $P$ time indices.
As $P$ increases, the tube radius decreases since $\sigma\in\mathcal{L}$, while a larger number of deviations outside the tube is permitted. Importantly, $P$ is independent of the dataset size $N$. Hence, for fixed $P$, the fraction of time indices at which the output lies outside the corresponding tube vanishes as $N\to\infty$, since~\eqref{eq:TP_card} directly implies that
\begin{equation}
	\lim_{N \to \infty}\frac{\#\mathcal{Q}(P,N)}{N+1} = 0,
\end{equation}
compare the discussion below Definition~\ref{def:intTP}.

The following result shows that the cumulative turnpike property implies the cardinality turnpike property.

\begin{prop}[Cumulative implies cardinality turnpike]\label{prop:intTP2cardTP}
	If the problem~\eqref{eq:NLP_case_1} has the cumulative turnpike property (Definition~\ref{def:intTP}), then it has the cardinality turnpike property (Definition~\ref{def:cardTP}).
\end{prop}
\begin{pf}
	Assume that problem \eqref{eq:NLP_case_1} has the cumulative turnpike property.
	Consider the function $\sigma:\mathbb{R}_{\geq0}\to\mathbb{R}_{\geq0}$ defined by
	\begin{equation}\label{eq:sigma_TP}
		\sigma(s):=
		\begin{cases}
			\alpha^{-1}\left({E}/{s}\right), & s\geq 1\\
			\alpha^{-1}({E}), & 0\leq s < 1
		\end{cases}
	\end{equation}
	with $E>0$ and $\alpha\in\mathcal{K}_\infty$ from Definition~\ref{def:intTP}. Then, $\sigma\in\mathcal{L}$.
	For the sake of contradiction, assume that the cardinality turnpike property does not hold for this choice of $\sigma$.
	This implies that there exists at least one solution $\theta'\in\mathcal{S}'(\bar{x})$ for $\bar{x}\in\mathcal{X}_0$ and a corresponding pair $(x^*_0,\theta^*)\in\mathcal{S}^*(\bar{x},\theta')$ such that more than $P$ elements satisfy $\|y_j(\bar{x},\theta',U^\mathrm{d})- y_j(x_0^*,\theta^*,U^\mathrm{d})\|>\sigma(P)$. Consequently, from~\eqref{eq:sigma_TP} and the fact that $P\geq 1$, we obtain
	\begin{equation*}
		\sum_{j=0}^{N}\alpha(\|y_j(\bar{x},\theta',U^\mathrm{d})-y_j(x_0^*,\theta^*,U^\mathrm{d})\|) > P\alpha(\sigma(P)) = E.
	\end{equation*}
	However, this contradicts~\eqref{eq:intTP}; thus, the cardinality turnpike property must hold for $\sigma$ chosen as in~\eqref{eq:sigma_TP}, completing the proof.
\end{pf}

While the cumulative turnpike always implies cardinality turnpike, the following counterexample reveals that the converse does not hold in general.

\begin{cexample}[Cardinality{\,}$\nRightarrow${\,}cumulative{\,}turnpike]\label{counterex}
	Consider the autonomous system
	\begin{equation*}
		\left\{
		\begin{matrix}
			w^+ = 2w\\
			m^+ = 0
		\end{matrix}
		\right.,
		\quad
		y =
		\begin{bmatrix}
			m(w+1) + (1-m)(\theta|w|-1)\\ (1-m)w\max\{2\theta|w|-1,0\}
		\end{bmatrix},
	\end{equation*}
	where the output is subject to an unknown parameter $\theta\in\Theta=[0,1]$.
	Define the state vector $x = (w,m)$ and consider the initial condition $x_0 = (w_0,m_0) \in\mathcal{X}_0= \{-1,1\}\times\{1\}$.
	Suppose that the output data is given by $y^\mathrm{d}_j = 0$, $j\in\mathbb{N}$. Consider the stage cost $l(y_1,y_2) = |[1,0](y_1-y_2)|$, regularization $r\equiv0$, and normalization weight $\beta_j = 1/(N+1)$, $j\in\mathbb{N}$.
	Then, the following statements apply:
	\begin{enumerate}[a)]
		\item The value functions in~\eqref{eq:NLP_case_TP} and~\eqref{eq:NLP_case_1} satisfy Assumption~\ref{ass:reachability}.
		
		\item The SEM problem~\eqref{eq:NLP_case_1} has the cardinality turnpike property (Definition~\ref{def:cardTP}).
		
		\item The SEM problem~\eqref{eq:NLP_case_1} does not have the cumulative turnpike property (Definition~\ref{def:intTP}).
	\end{enumerate}
\end{cexample}

\begin{pf}
	From the system dynamics and initial conditions, it follows that $w_j = 2^jw_0$, $j\geq0$, $m_0 = 1$, and $m_j = 0$, $j\in\mathbb{N}_+$; the corresponding outputs are $y_0(x_0,\theta) =
	(w_0 + 1, 0)$ and
	\begin{equation}\label{eq:outputs}
		y_j(x_0,\theta) = 
		\begin{bmatrix}
			\theta2^j-1\\2^jw_0\max\{2\theta2^j-1,0\}
		\end{bmatrix},\quad 
		j\in\mathbb{N}_+.
	\end{equation}
	Thus, the SEM cost function in~\eqref{eq:NLP_cost} specializes to
	\begin{equation}\label{eq:counterex:cost}
		J(x_0,\theta) = \frac{1}{N+1}\left(|w_0+1| + \sum_{j=1}^N |\theta 2^j-1|\right).
	\end{equation}
	We make the following claim.
	
	\begin{claim}\label{claim:minimizer}
		For each $N\in\mathbb{N}_{+}$, the unique minimizer of $J(x_0,\theta)$ in~\eqref{eq:counterex:cost} over $\mathcal{X}_0\times\Theta$ is $x_0^* = (-1,1)$ and $\theta^*=2^{-N}$.
	\end{claim}
	\begin{pf}
		The minimizer $x_0^*=(-1,1)$ is obvious.
		Since for fixed $N\in\mathbb{N}_{+}$, the factor $1/(N+1)$ does not affect the minimizer, it remains to show that $\theta=2^{-N}$ is the unique minimizer of $F(\theta) := \sum_{j=1}^N|\theta2^j-1|$ over $\Theta=[0,1]$, where $F(\theta)$ corresponds to the parameter-dependent part of the cost function $J(x_0^*,\theta)$.
		%, with its derivative given by $dF/d\theta = \sum_{j=1}^N\frac{2^j(\theta2^j - 1)}{\sqrt{(\theta2^j - 1)^2}}$.
		%For the sake of contradiction, suppose that $\theta^*\neq2^{-N}$.
		%For $\theta^*>2^{-N}$, note that $\frac{2^j(\theta2^j - 1)}{\sqrt{(\theta2^j - 1)^2}}$
		First, note that each summand $f_j(\theta):=|\theta2^j-1|$, $j\in[1,N]$ attains its unique minimum at $\theta=2^{-j}$ and is differentiable for $\theta\neq2^{-j}$, with
		\begin{equation}
			\frac{df_j}{d\theta}(\theta) =
			\begin{cases}
				-2^j, & \theta < 2^{-j} \\ 2^j, & \theta>2^{-j}
			\end{cases}, \qquad j\in[1,N].
		\end{equation}
		
		Consider $0<\theta<2^{-N}$. Since $2^{-N}\leq2^{-j}$ for all $j\in[1,N]$, it follows that ${df_j}/{d\theta}(\theta)<0$ for all $j\in[1,N]$.
		Consequently, since $F(\theta)$ is continuous on $[0,2^{-N}]$ and ${dF}/{d\theta}(\theta)<0$ on $(0,2^{-N})$, the mean value theorem implies that $F(\theta)$ is strictly decreasing on $[0,2^{-N}]$.
		
		Now consider $2^{-(N-k+1)}<\theta<2^{-(N-k)}$ for $k\in[1,N]$. We have $df_j/d\theta(\theta)=-2^j$ for $j\in[1,N-k]$, while $df_j/d\theta(\theta)=2^j$ for $j\in[N-k+1,N]$.
		Thus,
		\begin{align*}
			\frac{dF}{d\theta}(\theta)
			&= \sum_{j=N-k+1}^N2^j - \sum_{j=1}^{N-k}2^j \\
			&= \frac{2^{N-k+1}-2^{N+1}}{-1} - \frac{2-2^{N-k+1}}{-1}\\
			&= -2^{N-k+1}+2^{N+1} + 2-2^{N-k+1}\\
			&= 2^{N+1} -2^{N-k+2} + 2 > 0
		\end{align*}
		for all $k\in[1,N]$ and each $N\in\mathbb{N}_{+}$.
		Thus, $dF/d\theta(\theta)$ is strictly positive on each open interval $(2^{-(N-k+1)},2^{-(N-k)})$, $k\in[1,N]$. By continuity, $F(\theta)$ is therefore strictly increasing on the closed interval $[2^{-N},1]$.
		As $F(\theta)$ is strictly decreasing on $[0,2^{-N}]$, we can infer that $\theta=2^{-N}$ is the unique minimizer of $F(\theta)$ on $[0,1]=\Theta$, proving the claim.
	\end{pf}
	
	Due to the structure of the cost function in~\eqref{eq:counterex:cost}, for each $N\in\mathbb{N}_+$, the solution of the constrained SEM problem in~\eqref{eq:NLP_case_1} coincides with the optimal parameter $\theta^*$, that is, $\theta'=\theta^*$ for both choices of $\bar{x}\in\mathcal{X}_0$. For $N=0$, the cost is independent of $\theta$.
	
	To establish the cost reachability property in Assumption~\ref{ass:reachability}, we note that $V'(\bar{x}) - V^* = |\bar{w}+1|/(N+1)$ for $\bar{x} = (\bar{w},1) \in\mathcal{X}_0$. Using the definition of $\mathcal{X}_0$ and the fact that $M(N) = 1/(N+1)$ for each $N\in\mathbb{N}$ with $M(N)$ from~\eqref{eq:B_def}, we obtain that $V'(\bar{x}) - V^* \leq 2/(N+1) = 2M(N)$ uniformly for all $\bar{x}\in\mathcal{X}_0$ and all $N\in\mathbb{N}$.
	Thus, Assumption~\ref{ass:reachability} holds, which establishes the first statement.
	
	Let $y'_j := y_j(\bar{x},\theta')$ and $y^*_j := y_j({x}_0^*,\theta^*)$, $j\in[0,N]$.
	We now provide an estimate for $\|y'_j-y^*_j\|$ over $j\in[0,N]$.
	
	\begin{claim}\label{claim:ydiff}
		For all $\bar{x}\in\mathcal{X}_0$ and all $N\in\mathbb{N}$, we have
		\begin{equation}\label{eq:counterex:y_diff}
			\|y'_j-y^*_j\| \leq
			\begin{cases}
				2, & j = 0 \\
				0, & j \in [1,N-1] \\
				2^{N+1}, & j=N
			\end{cases}.
		\end{equation}
		Moreover, for the particular choice $\bar{x}=(1,1)$ and $N\in\mathbb{N}_+$,~\eqref{eq:counterex:y_diff} holds with equality at the terminal index $j=N$.
	\end{claim}
	\begin{pf}
		First, note that $\bar{x}\in\mathcal{X}_0$ implies that either $\bar{x}=(1,1)$ or $\bar{x}=(-1,1)=x_0^*$. Since $\theta'=\theta^*$, it therefore suffices to show that~\eqref{eq:counterex:y_diff} holds for the former choice of $\bar{x}$.

		For brevity, let $\bar{w}=1$ and $w^*=-1$. Consider $N\in\mathbb{N}_+$.
		At $j=0$, we have that $\|y'_0-y^*_0\| = |\bar{w}-w_0^*| = 2$.
		For $j\in[1,N-1]$, the max-term in \eqref{eq:outputs} evaluates to $\max\{2\theta^*2^j-1,0\}  = \max\{2^{-N+1+j}-1,0\} = 0$, where the first equality follows from Claim~\ref{claim:minimizer}. Consequently, $\|y'_j-y^*_j\| = 0$ for all $j\in[1,N-1]$.
		At the terminal index $j=N$, we have $\max\{2\theta^*2^j-1,0\} = \max\{1,0\} = 1$. Thus, $\|y'_N-y^*_N\| = 2^N(\bar{w}-w^*) = 2^{N+1}$.
		Combining the three cases proves the claim for $N\in\mathbb{N}_+$. For $N=0$, the bound follows directly from $y_0(x_0,\theta)=(w_0+1,0)$, independently of $\theta$.
	\end{pf}
	
	We now show that the cardinality turnpike property holds. To this end, let $\sigma(s) = 4/(s+1)$, $s\in\mathbb{R}_{\geq0}$. Then, $\sigma\in\mathcal{L}$. Moreover, for $P=1$, we have that $\sigma(P) = 2$. From~\eqref{eq:counterex:y_diff}, the set $\mathcal{Q}$ as defined in~\eqref{eq:TP_card_Q} contains at most the terminal index $j=N$, which verifies $\#\mathcal{Q}(P,N)\leq P$ for $P=1$.
	For any $P\geq 2$, \eqref{eq:counterex:y_diff} implies that $\mathcal{Q}$ contains at most two indices $j\in[0,N]$; namely, $j=0$ and $j=N$.
	Therefore, $\#\mathcal{Q}(P,N)\leq P$, for all $P\in\mathbb{N}_+$ and all $N\in\mathbb{N}$, and the cardinality turnpike property holds, establishing the second statement.
	
	Finally, consider $\bar{x}=(1,1)$. From Claim~\ref{claim:ydiff}, for every $\alpha\in\mathcal{K}_\infty$, it follows that $\sum_{j=0}^N\alpha(\|y'_j-y^*_j\|) \geq \alpha(\|y'_N-y^*_N\|)= \alpha(2^{N+1}) \to \infty$ for $N\to\infty$. Thus, the cumulative turnpike property does not hold.
	This establishes the third statement and hence completes the proof.
\end{pf}

\begin{rem}[Cardinality and cumulative turnpike]\label{rem:cardTP_converse}
	From Proposition~\ref{prop:intTP2cardTP} and Counterexample~\ref{counterex}, the cardinality turnpike property is strictly weaker than the cumulative turnpike property, even under Assumption~\ref{ass:reachability}.
	In fact, the essential difference is that cardinality turnpike bounds the number of exceptional time indices but imposes no restriction on the magnitude of the corresponding deviations. In contrast, the cumulative turnpike property also yields a uniform bound on every individual output discrepancy through~\eqref{eq:intTP}.
	Equivalence can be recovered under additional uniform boundedness assumptions, e.g., by restricting the identification problem to a compact output set $\mathcal{Y}$ via~\eqref{eq:NLP_X}, as shown in Proposition~\ref{prop:cardTP2intTP} below.
\end{rem}

\begin{prop}[Cumulative{\,}turnpike{\,}under{\,}compactness]
	\label{prop:cardTP2intTP}
	Assume that the output set $\mathcal{Y}$ in \eqref{eq:NLP_X} is compact.
	If problem~\eqref{eq:NLP_case_1} has the cardinality turnpike property (Definition~\ref{def:cardTP}), then it has the cumulative turnpike property (Definition~\ref{def:intTP}).
\end{prop}

\begin{pf}
	Consider any $N\in\mathbb{N}$, $\bar{x}\in\mathcal{X}_0$, $\theta'\in\mathcal{S}'(\bar{x})$, and $(x_0^*,\theta^*)\in\mathcal{S}^*(\bar{x},\theta')$ and define the discrepancies $d_j := \|y_j(\bar{x},\theta',U^\mathrm{d})-y_j(x_0^*,\theta^*,U^\mathrm{d})\|$, $j\in[0,N]$.
	Compactness of $\mathcal{Y}$ implies that $d_j\leq \bar{d}:= \max_{y_1,y_2\in\mathcal{Y}}\|y_1-y_2\|<\infty$ for all $j\in[0,N]$.
	
	Since $N$ is finite, there exists a bijection $\pi:[0,N]\rightarrow[0,N]$ such that the sequence $\{d_{\pi(j)}\}_{j=0}^N$ is non-increasing, i.e., $d_{\pi(j)}\geq d_{\pi(j+1)}$ for all $j\in[0,N-1]$.
	Under the cardinality turnpike property, we note that $d_{\pi(P)}\leq \sigma(P)$ for all $P\in[1,N]$.
	Indeed, if $d_{\pi(P)}>\sigma(P)$ for some $P\in[1,N]$,	then $d_{\pi(j)}>\sigma(P)$ for all $j\in[0,P]$, which yields $\#\mathcal{Q}(P,N)\geq P+1$ and contradicts \eqref{eq:TP_card}.
	
	Now define $\tilde{\sigma}(s) := \sigma(s) + 1/s$, $s> 0$. Since $\sigma\in\mathcal{L}$, $\tilde{\sigma}$ is continuous and strictly decreasing on $(0,\infty)$, with $\lim_{s \to 0^+}\tilde{\sigma}(s)=\infty$ and $\lim_{s \to \infty}\tilde{\sigma}(s)=0$. Hence, $\tilde{\sigma}:(0,\infty)\rightarrow(0,\infty)$ is bijective and its inverse $\tilde{\sigma}^{-1}$ is well-defined on $(0,\infty)$.
	Let $\alpha(s):={1}/{(\tilde{\sigma}^{-1}(s))^{2}}$, $s>0$, and $\alpha(0):=0$.
	Since $\tilde{\sigma}^{-1}$ is continuous, strictly decreasing, and satisfies $\lim_{s \to 0^+}\tilde{\sigma}^{-1}(s)=\infty$ and $\lim_{s \to \infty}\tilde{\sigma}^{-1}(s)=0$, it follows that $\alpha$ is of class $\mathcal{K}_\infty$.
	Moreover, $\alpha({\sigma}(s))<\alpha(\tilde{\sigma}(s)) = {1}/{(\tilde{\sigma}^{-1}(\tilde{\sigma}(s)))^{2}} = 1/s^2$ for all $s>0$. Consequently,
	$\sum_{s=1}^\infty\alpha({\sigma}(s)) \leq \sum_{s=1}^\infty{1}/{s^2} = {\pi^2}/{6} < \infty$.

	Combining the above estimates yields
	\begin{equation*}
	\begin{split}
		\sum_{j=0}^N\alpha(d_j) &= \sum_{j=0}^N\alpha(d_{\pi(j)})= \alpha(d_{\pi(0)}) + \sum_{P=1}^N\alpha(d_{\pi(P)})\\
		&\leq \alpha(\bar{d}) + \sum_{P=1}^N\alpha(\sigma(P))
		\leq \alpha(\bar{d}) + \pi^2/6 =:E,
	\end{split}
	\end{equation*}
	with $0<E<\infty$ being independent of $N\in\mathbb{N}$, $\bar{x}\in\mathcal{X}_0$, $\theta'\in\mathcal{S}'(\bar{x})$, and $(x_0^*,\theta^*)\in\mathcal{S}^*(\bar{x},\theta')$. Hence, Definition~\ref{def:intTP}
	is satisfied, which completes the proof.
\end{pf}

Hence, under compactness of \(\mathcal Y\), Propositions~\ref{prop:intTP2cardTP} and~\ref{prop:cardTP2intTP} imply that the cumulative and cardinality turnpike properties are, in fact, equivalent.
Notably, the proof of Proposition~\ref{prop:cardTP2intTP} can be directly linked to the mechanism underlying Counterexample~\ref{counterex}: while cardinality turnpike permits one exceptional output discrepancy to grow arbitrarily large with \(N\), compactness of \(\mathcal Y\) rules out precisely this possibility.

\section{Sufficient conditions for turnpike behavior}\label{sec:sysid_practice}

In this section, we provide sufficient conditions for the occurrence of turnpike behavior by exploiting the equivalence established in Theorem~\ref{thm:equivalences}.
Specifically, in Section~\ref{sec:reachability}, we establish cost reachability from an incremental output stability property, while in Section~\ref{sec:coerc}, we derive coercivity of the value function under suitable convexity and optimality conditions. Both results require a certain uniformity of the weights $\beta_j$ relative to their average $M(N)$; see Remark~\ref{rem:beta} for further discussion.

\subsection{Output stability implies cost reachability}
\label{sec:reachability}

We impose the following uniform continuity property on the stage cost~$l$.

\begin{ass}\label{ass:Lipschitz}
	The function $l:\mathbb{R}^p\times\mathbb{R}^p\rightarrow\mathbb{R}_{\geq0}$ in~\eqref{eq:NLP_cost} is locally Lipschitz continuous, i.e., for any compact sets $\mathcal{Y},\mathcal{Y}^\mathrm{d}\subset \mathbb{R}^p$, there exists $L>0$ such that for all $y_1,y_2\in\mathcal{Y}$, $\|l(y_1,y^\mathrm{d})-l(y_2,y^\mathrm{d})\| \leq L\|y_1-y_2\|$ uniformly for $y^\mathrm{d}\in\mathcal{Y}^\mathrm{d}$.
\end{ass}

Assumption~\ref{ass:Lipschitz} is a standard (mild) condition and is satisfied, in particular, for the practically relevant case of squared $L_2$-norm penalties.
Cost reachability occurs if the system  model additionally possesses the following stability property.

\begin{ass}[Incremental output stability]\label{ass:output_stability}
	The model in \eqref{eq:sys_1}--\eqref{eq:sys_2} is uniformly exponentially incrementally output stable over $\Theta$ along the training input sequence $U^\mathrm{d}$, that is, there exist constants $C>0$ and $\lambda\in(0,1)$ such that for all $N\in\mathbb{N}$,
	\begin{equation}
		\|y_j({x}_0,{\theta},U^\mathrm{d})-y_j(\hat{x}_0,{\theta},U^\mathrm{d})\| \leq C\lambda^j\|{x}_0-\hat{x}_0\| \label{eq:output_stability}
	\end{equation}
	for all $j\in[0,N]$, all $x_0,\hat{x}_0\in\mathcal{X}_0$, and all $\theta\in\Theta$.
\end{ass}

Assumption~\ref{ass:output_stability} essentially requires that the remote past of the learned process should be ``forgotten'' at an exponential rate. Thus, when being evaluated over the training input data, the identified model generates output sequences that converge to each other, despite different initial states.
Similar conditions are standard in the system identification literature \cite[Cond.~S3]{Ljung1978}, \cite[Cond.~1]{Beintema2023a}, although they are typically imposed on the true data-generating system.
In contrast, Assumption~\ref{ass:output_stability} concerns the learned model itself and may therefore be assessed after training or explicitly enforced during training; compare also \cite[Sec.~4.2]{Schiller2025b}.

\begin{prop}\label{prop:reachability}
	Let $\mathcal{Y}$ be compact and assume that any pair $(x_0,\theta)\in\mathcal{X}_0\times\Theta$ is feasible in~\eqref{eq:NLP}. Suppose that Assumptions~\ref{ass:Lipschitz} and \ref{ass:output_stability} are satisfied. Furthermore, suppose that there exists a constant $\overline{\beta}>0$ such that
	\begin{equation}\label{eq:beta_upper_bound}
		\beta_j \leq \overline{\beta} M(N), \quad j\in[0,N]
	\end{equation}
	uniformly for all $N\in\mathbb{N}$ with $M(N)$ from~\eqref{eq:B_def}.
	Then, Assumption \ref{ass:reachability} holds.
\end{prop}

\begin{pf}	
	Consider the value functions from~\eqref{eq:NLP_case_TP} and~\eqref{eq:NLP_case_1}. By optimality, these obey the following relations:
	\begin{equation}\label{eq:proof_optimality}
		V^* = J(x_0^*,\theta^*) \leq V'(\bar{x}) = J(\bar{x},\theta') \leq J(\bar{x},\theta^*)
	\end{equation}
	for any choice of $(x_0^*,\theta^*)\in\mathcal{S}^*$, $\theta'\in\mathcal{S}'(\bar{x})$, and $\bar{x}\in\mathcal{X}_0$.
	Applying the triangle inequality, non-negativity of $\beta_j$, Assumption~\ref{ass:Lipschitz}, the uniform upper bound of $\beta_j$ in~\eqref{eq:beta_upper_bound}, Assumption~\ref{ass:output_stability}, and the geometric series leads to
	\begin{align}
		&V'(\bar{x}) - V^* \leq \|J(\bar{x},\theta^*) - J(x_0^*,\theta^*)\|\nonumber\\
		&\leq  \sum_{j=0}^N \beta_j\left\| l({y}_j(\bar{x},\theta^*,U^\mathrm{d}),{y}_{j}^\mathrm{d}) - l({y}_j(x_0^*,\theta^*,U^\mathrm{d}),{y}_{j}^\mathrm{d}) \right\|\nonumber\\
		&\leq M(N)\overline{\beta} L  \sum_{j=0}^N \left\|{y}_j(\bar{x},\theta^*,U^\mathrm{d})-{y}_j(x_0^*,\theta^*,U^\mathrm{d})\right\|\nonumber\\
		&\leq M(N) \frac{\overline{\beta}LC}{1-\lambda} \|\bar{x}-x_0^*\|. \label{eq:proof_cost_reachability}
	\end{align}
	As $\mathcal{X}_0$ is compact, there exists $C_\mathcal{X} := \max_{x_1,x_2\in\mathcal{X}_0}\|x_1-x_2\|$ with $C_\mathcal{X}<\infty$. Thus, letting $E_\mathrm{r} = \overline{\beta}{LCC_\mathcal{X}}{/(1-\lambda)}$ in \eqref{eq:proof_cost_reachability} establishes~\eqref{eq:reachability}, completing the proof.
\end{pf}

\subsection{Optimality implies coercivity}\label{sec:coerc}

In the following, we establish coercivity of the value function under suitable conditions.
Here, we consider $r\equiv0$ for simplicity and restrict ourselves to convex stage costs $l$.

\begin{ass}[Uniform strict convexity]\label{ass:convex_lr}
	The function $l$ is uniformly strictly convex with modulus $\phi\in\mathcal{K}_\infty$, i.e., it holds that
	\begin{equation*}
		\begin{split}
			l(ty_1+(1-t)y_2,y^\mathrm{d})
			\leq&\ tl(y_1,y^\mathrm{d})+(1-t)l(y_2,y^\mathrm{d})\nonumber\\
			&\ -t(1-t)\phi(\|y_1-y_2\|)
		\end{split}
	\end{equation*}
	for all $y_1,y_2\in\mathcal{Y}$ uniformly for all $y^\mathrm{d}\in\mathcal{Y}^\mathrm{d}$ and $t\in[0,1]$.
\end{ass}

Note that Assumption~\ref{ass:convex_lr} is immediately satisfied if $l$ is strongly convex (which applies, e.g., for squared $L_2$-norm penalties), yielding $\phi(s)={\kappa}s^2$ for some parameter $\kappa>0$.
We now state the required optimality condition.

\begin{ass}[Optimality condition]\label{ass:optimality}
	There exists a constant $K\geq0$ such that, for all $N\in\mathbb{N}$ and $\bar{x}\in\mathcal{X}_0$, it holds that
	\begin{equation}\label{eq:cond_V_Jt}
		V^* - J(t) \leq M(N)K, \ \ \ J(t) := \sum_{j=0}^N\beta_jl(y_j(t),y^\mathrm{d}_j), \ \ t\in[0,1]
	\end{equation}
	for all $\theta'\in\mathcal{S}'(\bar{x})$ and all $(x_0^*,\theta^*)\in\mathcal{S}^*(\bar{x},\theta')$, where $y_j(t) = ty_j(\bar{x},\theta',U^\mathrm{d}) + (1-t)y_j(x_0^*,\theta^*,U^\mathrm{d})$, $j\in[0,N]$, $t\in[0,1]$.
\end{ass}

Before stating the main result of this section, we further discuss Assumption~\ref{ass:optimality} and provide some intuition on its interpretation and applicability.

\begin{rem}[Optimality]\label{rem:optimality}
	Assumption~\ref{ass:optimality} requires the difference of the unconstrained value function $V^*$ and the cost achieved by point-wise convex combinations of outputs generated by $(\bar{x},\theta')$ and the corresponding optimal pair $(x_0^*,\theta^*)$ to be uniformly bounded. Since the stage cost is nonnegative, a sufficient condition is $K_V:=\sup_{N\in\mathbb{N}}\,{V^*}/{M(N)} < \infty$, in which case \eqref{eq:cond_V_Jt} holds with $K=K_V$.
	This applies, for example, under exact realizability (for which \(V^*=0\)) and, for quadratic stage costs with weights satisfying an upper comparability condition as in~\eqref{eq:beta_upper_bound}, under finite-energy disturbances or exponentially decaying residuals.
	Such settings arise, in particular, for synthetic data and overparameterized deep learning models; see, e.g., \cite{Mhaskar2020}.
\end{rem}

\begin{rem}[Optimality under quadratic costs]\label{rem:optimality_quadratic}
	\, For the quadratic stage cost $l(y,y^{\mathrm d})=\|y-y^{\mathrm d}\|^2$, the optimality condition in Assumption~\ref{ass:optimality} admits a more explicit interpretation.
	To see this, let $r\equiv0$ for simplicity, assume that $M(N)>0$, and define the normalized weights $\omega_j:=\beta_j/M(N)$, $j\in[0,N]$.
	Fix some $\bar{x}\in\mathcal{X}_0$ and consider arbitrary $\theta'\in\mathcal{S}'(\bar{x})$ and $(x_0^*,\theta^*)\in\mathcal{S}^*(\bar{x},\theta')$, generating the outputs $y_j' = y_j(\bar{x},\theta',U^{\mathrm d})$ and $y_j^*=y_j(x_0^*,\theta^*,U^{\mathrm d})$, $j\in[0,N]$.
	Let $e_{j}:=y_j^*-y_j^{\mathrm d}$ and $\Delta y_{j}:=y_j'-y_j^*$ and define
	\begin{equation*}
		A_N:=\sum_{j=0}^N\omega_{j}\langle e_{j},\Delta y_{j}\rangle,
		\quad
		B_N:=\sum_{j=0}^N\omega_{j}\|\Delta y_{j}\|^2.
	\end{equation*}
	Then, the cost $J(t)$ in~\eqref{eq:cond_V_Jt} satisfies%
	\begin{equation*}
		J(t) = V^* + M(N)(2tA_N+t^2B_N),
		\quad t\in[0,1].
	\end{equation*}	
	Consequently,
	\begin{equation*}
		V^*-J(t) \leq V^*-\min_{t\in[0,1]}J(t) = M(N)\max_{t\in[0,1]} (-2tA_N-t^2B_N)
	\end{equation*}
	If $A_N\geq0$, the maximum equals zero.
	Suppose that $A_N<0$. Since $V'(\bar{x})=J(1)\geq V^*$, it follows that $2A_N+B_N\geq0$, and hence, $B_N\geq -2A_N>0$. The maximum is therefore attained at $t=-A_N/B_N\in(0,1]$ and yields $\max_{t\in[0,1]} \bigl(-2tA_N-t^2B_N\bigr) = A_N^2/B_N$, whenever $A_N<0$.
	Thus, Assumption~\ref{ass:optimality} becomes equivalent to the existence of a uniform constant $K\geq0$ such that ${A_N^2}/{B_N}\leq K$ whenever $A_N<0$.
	Since $B_N\geq-2A_N$ implies ${A_N^2}/{B_N}\leq -{A_N}/{2}$, a sufficient condition is that the negative parts $A_N^-:=\max\{-A_N,0\}$ are uniformly bounded. In particular, this is satisfied if $A_N^-\to0$ uniformly over the considered optimal solutions as $N\to\infty$.
	
	The quantity $A_N$ has a direct first-order interpretation, as $\left.\frac{\mathrm d}{\mathrm dt}\frac{J(t)-V^*}{M(N)}\right|_{t=0}=2A_N$.
	Thus, $A_N<0$ means that moving from the unconstrained toward the constrained optimal output sequence via the convex combinations considered in Assumption~\ref{ass:optimality} initially decreases the quadratic cost, whereas $A_N^-\to0$ means that this possible first-order decrease vanishes asymptotically.
	This is closely related to the classical first-order conditions for least-squares estimation: at an optimum, the residual is orthogonal to output variations induced by parameter perturbations---globally for linear models and locally, through linearization at the optimum, for nonlinear models; see, e.g.,~\cite{Bates1988,Ljung1999}.
	In particular, for affine least-squares problems with a convex feasible set, first-order optimality yields $A_N\geq0$ (and $A_N=0$ for an interior optimum); see \cite[Sec.~4.2.3]{Boyd2004}. For nonlinear models, such point-wise convex combinations of trajectories are generally not realizable by the model, and hence $A_N<0$ may occur despite optimality. In this case, uniform boundedness of \(A_N^-\) follows, e.g., under additional regularity conditions ensuring that the constrained and unconstrained SEM solutions remain uniformly close in the sense that \(B_N\) is uniformly bounded, since \(A_N^-\le B_N/2\); compare also \cite[Sec.~IV]{Ljung1978}.
\end{rem}

With these observations, we can now establish coercivity under Assumptions~\ref{ass:convex_lr} and~\ref{ass:optimality}.

\begin{thm}\label{thm:coerc}	
	Let Assumptions~\ref{ass:convex_lr} and~\ref{ass:optimality} be satisfied.
	Suppose that there exists a constant $\underline{\beta}>0$ such that%
	\begin{equation}\label{eq:beta_lower_bound}
		\beta_j \geq \underline{\beta}M(N), \quad j\in[0,N]
	\end{equation}
	uniformly for all $N\in\mathbb{N}$ with $M(N)$ from~\eqref{eq:B_def}.
	Then, the value function in~\eqref{eq:NLP_case_1} is coercive in the sense of Definition~\ref{def:coerc}.
\end{thm}

\begin{pf}
	Define ${l}_\Delta(y_1,y_2,y^\mathrm{d}):= l(y_1,y^\mathrm{d}) - l(y_2,y^\mathrm{d})$ for $y_1,y_2,y^\mathrm{d}\in\mathbb{R}^p$.
	Consider arbitrary $y_1,y_2\in\mathcal{Y}$ and $y^\mathrm{d}\in\mathcal{Y}^\mathrm{d}$ and introduce $y(t) := ty_1 + (1-t)y_2$, $t\in[0,1]$.
	From Assumption~\ref{ass:convex_lr}, we have
	\begin{align*}
		&l_\Delta(y(t),y_2,y^\mathrm{d}) = l(y(t),y^\mathrm{d}) - l(y_2,y^\mathrm{d}) \\
		&\leq tl(y_1,y^\mathrm{d})+(1-t)l(y_2,y^\mathrm{d}) \\
		&\quad - t(1-t)\phi(\|y_1-y_2\|) - l(y_2,y^\mathrm{d})\\
		&= tl(y_1,y^\mathrm{d})-tl(y_2,y^\mathrm{d}) - t(1-t)\phi(\|y_1-y_2\|)\\
		&= tl_\Delta(y_1,y_2,y^\mathrm{d}) - t(1-t)\phi(\|y_1-y_2\|).
	\end{align*}
	Rearranging and dividing by $t\in(0,1)$ yields
	\begin{equation*}
		l_\Delta(y_1,y_2,y^\mathrm{d})  \geq \frac{1}{t}l_\Delta(y(t),y_2,y^\mathrm{d}) + (1-t)\phi(\|y_1-y_2\|).
	\end{equation*}
	Using the definition of $l_\Delta$ leads to
	\begin{equation}
		\begin{split}\label{eq:convexity_res1}
		l(y_1,y^\mathrm{d}) - l(y_2,y^\mathrm{d})
		\geq &\ \frac{1}{t}l(y(t),y^\mathrm{d}) - \frac{1}{t}l(y_2,y^\mathrm{d}) \\
		&+(1-t)\phi(\|y_1-y_2\|).
		\end{split}
	\end{equation}

	Consider any $N\in\mathbb{N}$, $\bar{x}\in\mathcal{X}_0$, $\theta'\in\mathcal{S}'(\bar{x})$, and $(x_0^*,{\theta}^*)\in\mathcal{S}^*(\bar{x},\theta')$, generating the outputs $y'_j:=y_j(\bar{x},{\theta}',U^\mathrm{d})$ and ${y}_j^*:=y_j({x}^*_0,{\theta}^*,U^\mathrm{d})$, $j\in[0,N]$. Moreover, for each $j\in[0,N]$, let ${y}_j(t) := ty'_j + (1-t)y^*_j$, $t\in[0,1]$. From~\eqref{eq:convexity_res1} and Assumption~\ref{ass:optimality}, we can immediately conclude that
	\begin{equation}
		V'(\bar{x}) - V^* \geq -M(N)\frac{K}{t} + (1-t)\sum_{j=0}^N\beta_j\phi(\|y'_j-{y}^*_j\|)\label{eq:proof_convexity}
	\end{equation}
	with $t\in(0,1)$.
	Now arbitrarily fix $t\in(0,1)$ and let $c:=1-t\in(0,1)$ and $\bar{K}:=K/t\geq0$. Consequently, from~\eqref{eq:proof_convexity} and the lower bound in~\eqref{eq:beta_lower_bound}, we can infer that
	\begin{equation*}
		V'(\bar{x}) - V^* \geq -M(N)\bar{K} +  M(N)\underline{\beta}c\sum_{j=0}^N\phi(\|y'_j-{y}^*_j\|)\label{eq:proof_convexity_2}
	\end{equation*}
	uniformly for $N\in\mathbb{N}$, $\bar{x}\in\mathcal{X}_0$, $\theta'\in\mathcal{S}'(\bar{x})$, and $(x_0^*,{\theta}^*)\in\mathcal{S}^*(\bar{x},\theta')$.
	Setting $\gamma(s) = \underline{\beta}c\cdot s$, $\alpha_\mathrm{c} = \phi$, and $C_\mathrm{c}= \bar{K}$ establishes the coercivity property from Definition~\ref{def:coerc}, thus completing the proof.
\end{pf}

\begin{rem}[Uniform weight comparability]\label{rem:beta}
	The bounds in~\eqref{eq:beta_upper_bound} and \eqref{eq:beta_lower_bound} ensure that the weights $\beta_j$ do not become negligible or dominant relative to their average.
	This is satisfied for standard system identification settings, e.g., when using an unweighted cost function (where $\beta_j\equiv1$) or a normalized one (where $\beta_j\equiv 1/(N+1)$).
	In both cases, conditions~\eqref{eq:beta_upper_bound} and \eqref{eq:beta_lower_bound} hold with $\underline{\beta}=\overline{\beta}=1$.
	In contrast, the bounds in~\eqref{eq:beta_upper_bound} and \eqref{eq:beta_lower_bound} cannot, in general, be satisfied uniformly in $N$ for, e.g., an exponentially discounted cost using $\beta_j=\rho^{N-j}$, $j\in[0,N]$ for some discount factor $\rho\in(0,1)$.
\end{rem}

\begin{rem}[Incorporating zero weights]\label{rem:burn_in}
	Theorem~\ref{thm:coerc} can be easily extended to the case where the weightings $\beta_j$ are selected such that $\beta_j=0$ for a subset of indices $j\in[0,N]$. This allows one to, e.g., reject certain outliers in the cost function or to include a burn-in phase in order to reduce the influence of model transients; compare \cite{Schiller2025b,Jaeger2002,Bonassi2022} and see the numerical example in Section~\ref{sec:example}.
	Specifically, when imposing \eqref{eq:beta_lower_bound} only for $j\in\mathcal{I}:=\{j\in[0,N] \mid \beta_j>0\}$, it suffices that the cardinality of the zero-weight index set $\mathcal{G} := [0,N]\setminus \mathcal{I}$ remains uniformly bounded with respect to $N$, i.e., $\sup_{N\in\mathbb N}\#\mathcal{G}<\infty$.
	This ensures that only a uniformly bounded number of weights may vanish, so that the cost function does not ignore a growing portion of the trajectory when $N$ is increasing.
	Then, under compactness of $\mathcal{Y}$, one can suitably modify the proof of Theorem~\ref{thm:coerc} (specifically, \eqref{eq:proof_convexity}) and establish the coercivity property from Definition~\ref{def:coerc}, using the general fact that, for all $N\in\mathbb{N}$,
	\begin{equation*}
		\sum_{j\in\mathcal{I}} a_j=  \sum_{j=0}^N a_j - \sum_{j\in\mathcal{G}} a_j\geq \sum_{j=0}^N a_j -\#\mathcal{G}\cdot \max_{j\in\mathcal{G}}a_j
	\end{equation*}
	for summands $0\leq a_j<\infty$, $j\in[0,N]$.
\end{rem}

\section{Numerical example}\label{sec:example}

To illustrate the theory, we consider the scalar nonlinear state-space model
\begin{equation}
	x^+ = \tanh(\theta x + u), \qquad y=x, \label{eq:ex_sys}
\end{equation}
where $\theta$ is the parameter to be identified.
The model~\eqref{eq:ex_sys} can equivalently be interpreted as a scalar Elman-type recurrent neural network with one hidden unit, recurrent weight $\theta$, fixed input weight equal to one, zero bias, and an identity readout map. It therefore provides a simple representative of the recurrent neural network models motivating the use of truncated SEM, while retaining sufficient analytical tractability to explicitly characterize the relevant optimal solutions. For larger benchmark examples illustrating in particular the role of a burn-in phase in recurrent neural network training using the SEM method, the reader is referred to~\cite{Schiller2025b}.

We consider the compact sets $\mathcal{X}_0 = [-1,1]$ and $\Theta = [-0.9,0.9]$ and generate a training dataset $D$ as in~\eqref{eq:Di} by simulating the model~\eqref{eq:ex_sys} using the true unknown parameter $\theta^\mathrm{true} = -0.4$, initial condition $x_0^\mathrm{true}=0.6$, and a sinusoidal input trajectory $U^\mathrm{d}$.
The absence of measurement noise is deliberate: it isolates the effect caused by fixing the initial condition and permits a direct verification of the optimality condition in Assumption~\ref{ass:optimality}.

We use the quadratic stage cost $l(y_1,y_2) = (y_1 - y_2)^2$ with $r \equiv 0$ and incorporate the following weighting scheme\footnote{
	For the considered weighting scheme, we in fact obtain $M(N)=0$ for $N\in\{0,1\}$ and $M(N)>0$ for all $N\geq2$. We therefore restrict the following analysis, where necessary, to $N\geq2$. The degenerate cases $N\in\{0,1\}$, for which $J\equiv0$, can be treated separately and do not affect the turnpike conclusions.
} corresponding to a fixed burn-in phase of length $m=2$, i.e., $\beta_j = 0$, $j\in\{0,1\}$ and $\beta_j = 1$, $j\in[2,N]$, compare Remark~\ref{rem:burn_in}.
The initial state of the constrained SEM problem is fixed to $\bar{x}=-0.6$.

\paragraph{Verification of the theoretical conditions.}
The model in~\eqref{eq:ex_sys} is uniformly incrementally output stable over $\Theta$. Indeed, using the fact that $\tanh$ is globally Lipschitz continuous with Lipschitz constant one, we obtain
\begin{equation*}
	|y_j(x_0,\theta,U^{\mathrm d}) - y_j(\hat{x}_0,\theta,U^{\mathrm d})|
	\leq |\theta|^j|x_0-\hat{x}_0| \leq 0.9^j|x_0-\hat{x}_0|
\end{equation*}
for all $j\in\mathbb{N}$, all $x_0,\hat{x}_0\in\mathcal{X}_0$, and all $\theta\in\Theta$.
Hence, Assumption~\ref{ass:output_stability} holds with $C=1$ and $\lambda=0.9$.

Furthermore, $x_j\in(-1,1)$ for all $j\geq1$, independently of $x_0\in\mathcal{X}_0$.
The model outputs $y_j$ therefore remain in the compact set $\mathcal{Y}:=\mathcal{X}:=\mathcal{X}_0$ for all $j\in\mathbb{N}$. Consequently, every pair $(x_0,\theta)\in\mathcal{X}_0\times\Theta$ is feasible in~\eqref{eq:NLP}. Moreover, the stage cost satisfies Assumptions~\ref{ass:Lipschitz} and~\ref{ass:convex_lr}, with $\phi(s)=s^2$.
For $j\in[2,N]$, the weights $\beta_j$ are uniformly comparable to $M(N)$ and satisfy $M(N)\leq\beta_j\leq3M(N)$, while the burn-in phase yields $\beta_j=0$ for $j\in\{0,1\}$. Hence, \eqref{eq:beta_upper_bound} holds with $\overline{\beta}=3$, while the extension of \eqref{eq:beta_lower_bound} in Remark~\ref{rem:burn_in} applies with $\underline{\beta}=1$, due to the fact that the resulting zero-weight index set satisfies $\mathcal{G}=\{0,1\}$ and hence $\#\mathcal{G}=2$ uniformly for all $N\geq2$.
Since the dataset $D$ is generated by an admissible pair, we have $J(x_0^{\mathrm{true}},\theta^{\mathrm{true}}) = 0$ for every integer $N\geq2$. Non-negativity of the cost therefore implies that $V^*=0$ and, according to Remark~\ref{rem:optimality}, Assumption~\ref{ass:optimality} holds with $K=0$.
Consequently, Proposition~\ref{prop:reachability} and Theorem~\ref{thm:coerc} establish cost reachability and coercivity, respectively. Application of Theorem~\ref{thm:equivalences} then implies that the constrained SEM problem~\eqref{eq:NLP_case_1} is strictly dissipative and has the cumulative turnpike property.

\paragraph{Non-unique optimal solutions for $N{\,=\,}m$.}
We first consider the special case where $N=m=2$, for which only the terminal output at $j=N$ is penalized. Hence, the cost function~\eqref{eq:NLP_cost} reduces to $J(x_0,\theta) =(y_2(x_0,\theta,U^{\mathrm d})-y_2^{\mathrm d})^2$, where $y_2(x_0,\theta,U^{\mathrm d}) = \tanh(\theta \tanh(\theta x_0 + u^\mathrm{d}_0) + u^\mathrm{d}_1)$ using the model equations from~\eqref{eq:ex_sys}.
Since $V^*=0$, any solution $(x_0^*,\theta^*)$ must yield $y_2(x_0^*,\theta^*,U^\mathrm{d}) = y^\mathrm{d}_2$. Thus, for a fixed $\theta\in\Theta\setminus\{0\}$, this is satisfied for
\begin{equation*}
	x_0^*(\theta) := \frac{\artanh(x_1^*(\theta))-u^\mathrm{d}_0}{\theta},
	\quad
	x_1^*(\theta) := \frac{\artanh(y^\mathrm{d}_2)-u^\mathrm{d}_1}{\theta}.
\end{equation*}
Consequently, the unconstrained SEM problem~\eqref{eq:NLP_case_TP} possesses the non-singleton solution set
\begin{equation*}
	\mathcal{S}^*{=}\left\{(x_0^*(\theta),\theta) \mid \theta{\,\in\,}\Theta{\,\setminus}\{0\}, x_1^*(\theta){\,\in\,}(-1,1), x_0^*(\theta){\,\in\,}\mathcal{X}_0\right\}\hspace{-0.4ex}.
\end{equation*}
This set forms a nonlinear curve in the $(\theta,x_0)$-plane, which is illustrated by the solid lines in the middle panel in Figure~\ref{fig:example}.

We now consider the constrained SEM problem in~\eqref{eq:NLP_case_1} with fixed initial condition $\bar{x}=-0.6$. The left panel in Figure~\ref{fig:example} shows the cost function $J(\bar{x},\theta)$ for admissible $\theta\in\Theta$, from which we find that $\mathcal{S}'(\bar{x}) = \{\theta'\}$ with $\theta' =  0.152$.
Since the distance measure $D$ in~\eqref{eq:J_tilde} only contains the weighted terminal output for $N=m=2$ and every pair in $\mathcal{S}^*$ generates the same terminal output $y_2^{\mathrm d}$, all elements of $\mathcal{S}^*$ are equally close to the constrained solution. Thus, $\mathcal{S}^*(\bar{x},\theta')=\mathcal{S}^*$.
The example therefore explicitly exhibits the non-unique turnpike setting addressed by Definition~\ref{def:intTP}.

\paragraph{Increasing the dataset size.}
We now increase the dataset size and additionally consider $N\in\{20,200,2{,}000,20{,}000\}$.
For every $N>m$, the true pair $(x_0^{\mathrm{true}},\theta^{\mathrm{true}})$ remains an optimal solution with zero cost---in fact, it is unique. Indeed, every optimal pair necessarily satisfies $y_2=y_2^{\mathrm d}$ and $y_3=y_3^{\mathrm d}$.
Using the model dynamics~\eqref{eq:ex_sys} and injectivity of $\tanh$, it follows that, provided $y_2^{\mathrm d}\neq0$, $\theta=({\artanh(y_3^{\mathrm d})-u_2^{\mathrm d}})/{y_2^{\mathrm d}} = \theta^{\mathrm{true}}\neq0$.
The initial condition is then uniquely recovered by backward recursion, yielding $x_0=x_0^{\mathrm{true}}$.
Note that both $y_2^{\mathrm d}\neq0$ and $\theta^{\mathrm{true}}\neq0$ hold for the considered dataset.
Consequently, $\mathcal{S}^*=\{(x_0^{\mathrm{true}},\theta^{\mathrm{true}})\}$ for each $N\geq3$, depicted by the blue $\times$-marker in the middle plot of Figure~\ref{fig:example}.

\begin{table}[t]
	\centering
	\caption{Constrained SEM solutions, scaled value-function gaps, and cumulative output discrepancies for varying dataset sizes.}
	\label{tab:example_results}
	\begin{tabular*}{\tblwidth}{@{}LCCC@{}}
		\toprule
		$N$
		& $\theta'$
		& $\frac{V'(\bar{x})-V^*}{M(N)}\cdot10^2$
		& $E_N$\\
		\midrule
		2 & \phantom{-}0.152 & 0.003 & 2.587 \\
		20 & -0.322 & 2.464 & 1.641 \\
		200 & -0.383 & 3.447 & 1.675 \\
		2,000 & -0.398 & 3.665 & 1.684 \\
		20,000 & -0.400 & 3.691 & 1.685 \\
		\bottomrule
	\end{tabular*}
\end{table}

\begin{figure*}
	\vspace{1.2ex}
	\includegraphics{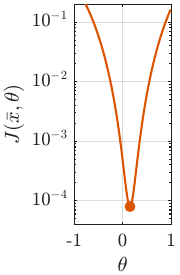}\hfill
	\includegraphics{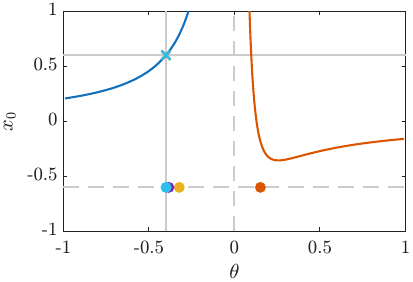}\hfill
	\includegraphics{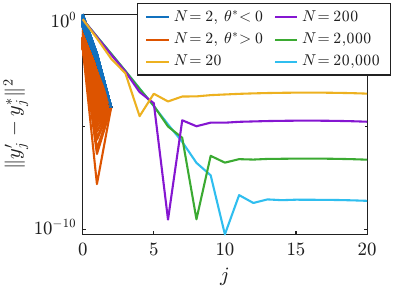}
	\caption{
		Left: Cost function $J(\bar{x},\theta)$ evaluated over $\theta{\,\in\,}\Theta$ using $N{\,=\,}m{\,=\,}2$.
		Middle: Optimal solution sets $\mathcal{S}^*$ (solid curves and $\times$-marker) and constrained solutions $\mathcal{S}'(\bar{x})$ (dot-markers) in the $(\theta,x_0)$-plane for different dataset sizes $N$; the colors correspond to those in the right panel.
		Right: Squared output discrepancies over time. For $N{\,=\,}2$, each curve corresponds to an optimal pair $(x_0^*,\theta^*)$ sampled from the non-singleton solution set $S^*$.
	}
	\label{fig:example}
\end{figure*}

For each considered dataset size $N$, we solve the constrained SEM problem~\eqref{eq:NLP_case_1} numerically in MATLAB using the \texttt{fmincon} routine. The obtained constrained solutions are unique and are denoted by $\mathcal{S}'(\bar{x})=\{\theta'\}$.
The corresponding constrained parameter estimates $\theta'$, scaled value-function gaps $({V'(\bar{x})-V^*})/{M(N)}$, and cumulative output discrepancies $E_N$ are reported in Table~\ref{tab:example_results}, where\footnote{\label{footnote:EN}
	For $N=2$, all elements of $\mathcal S^*$ are equally close to the constrained solution, so that $\mathcal S^*(\bar x,\theta')=\mathcal S^*$. Since Definition~\ref{def:intTP} requires the cumulative bound to hold for all associated closest optimal solutions, $E_N$ denotes the corresponding worst-case left-hand side of~\eqref{eq:intTP}. For $N\geq3$, the set $\mathcal S^*$ reduces to a singleton.
	We deliberately consider squared discrepancies, as, for the present setup, the theory establishes the cumulative turnpike property with \(\alpha(s)=s^2\): The proof of Theorem~\ref{thm:coerc} yields $\alpha_c(s)=\phi(s)=s^2$ in~\eqref{eq:coerc}, while the proof of Lemma~\ref{lem:coerc2intTP} yields $\alpha=\alpha_c$ in~\eqref{eq:intTP}.
}
\begin{equation*}
	E_N{:=}\max_{(x_0^*,\theta^*)\in\mathcal{S}^*(\bar{x},\theta')} \ \sum_{j=0}^N \big(y_j(\bar{x},\theta',U^{\mathrm d}) - y_j(x_0^*, \theta^*, U^{\mathrm d})\big)^2.
\end{equation*}
The associated solutions are additionally depicted in the middle panel of Figure~\ref{fig:example}.
From these results, we observe that the estimates $\theta'$ approach the true parameter $\theta^{\mathrm{true}}=-0.4$ as the dataset size increases.
Moreover, the scaled value-function gap remains bounded over the considered dataset sizes and appears to converge to a finite value. Its boundedness is consistent with the cost reachability property in Assumption~\ref{ass:reachability}.
Most importantly, the cumulative output discrepancy $E_N$ remains uniformly bounded despite the increasing dataset size and approaches approximately $1.685$ for the largest considered size.
This boundedness is precisely the behavior captured by the cumulative turnpike property in Definition~\ref{def:intTP}: for the present setup, $E_N$ coincides with the worst-case value of the left-hand side in~\eqref{eq:intTP} for the function $\alpha(s)=s^2$ established by the theory; see Footnote~\ref{footnote:EN} for details. Hence, Theorem~\ref{thm:equivalences} directly guarantees the observed uniform boundedness of $E_N$.

Figure~\ref{fig:example} provides a graphical illustration of these results.
The right panel shows the squared output discrepancies over time for the different dataset sizes.
The discrepancies are predominantly confined to an initial portion of the trajectory and remain small thereafter, progressively becoming smaller for increasing values of $N$. Together with the boundedness properties reported in Table~\ref{tab:example_results}, this behavior is consistent with the cumulative turnpike property in Definition~\ref{def:intTP}.

\section{Conclusion}
In this paper, we studied turnpike behavior in SEM problems for system identification.
Such behavior is desirable in practice, as it allows one to employ computationally more tractable constrained SEM problems with a fixed initial state while ensuring that the resulting output sequence remains close to an optimal output sequence of the corresponding unconstrained problem.

We generalized the notion of cumulative output turnpike to the case of non-unique optimal output sequences and established its equivalence to a coercivity property of the value function and a tailored notion of strict dissipativity. Moreover, we introduced a cardinality turnpike property and showed that it is strictly weaker than the cumulative notion, while equivalence can be recovered under additional uniform boundedness assumptions.
We further derived sufficient conditions for turnpike behavior in system identification based on incremental output stability, convexity of the stage cost, and a suitable optimality condition.
Finally, we considered a simple Elman-type recurrent neural network for which these conditions can be verified analytically and illustrated the resulting turnpike behavior numerically.

{Future work might include extending the analysis to more general SEM formulations augmented with additional encoders and investigating relaxed conditions under which turnpike behavior can be rigorously guaranteed.}

%%%%%%%%%%%%%%%%%%%%%%%%%%%%%%%%%%%%%%%%%%%%%%%%%%%%%%%%%%%%%%%%%%%%%%%%%%%%%%%%%
\appendix

\section{Proof of Theorem~\ref{thm:equivalences}}
\label{sec:thm_equivalences}

We prove Theorem~\ref{thm:equivalences} by establishing the implications
\begin{itemize}
	\itemsep=-0.4ex
	\item coercivity $\Rightarrow$ cumulative turnpike (Lemma~\ref{lem:coerc2intTP}),
	\item cumulative turnpike $\Rightarrow$ strict dissipativity (Lemma~\ref{lem:intTP2diss}),
	\item strict dissipativity $\Rightarrow$ coercivity (Lemma~\ref{lem:diss2coerc}).
\end{itemize}

\begin{lem}[Coercivity implies cumulative turnpike]\label{lem:coerc2intTP}
	Let Assumption~\ref{ass:reachability} hold. If the value function in~\eqref{eq:NLP_case_1} is coercive (Definition~\ref{def:coerc}), then the problem~\eqref{eq:NLP_case_1} has the cumulative turnpike property (Definition~\ref{def:intTP}).
\end{lem}
\begin{pf}
	Assuming that $M(N)>0$ for all $N\in\mathbb{N}$, the combination of Definition~\ref{def:coerc} and Assumption~\ref{ass:reachability} leads to
	\begin{equation}\label{eq:proof_cumulative}
		\gamma\left(\sum_{j=0}^{N} \alpha_\mathrm{c}(\|y'_j-y^*_j\|)\right) - C_\mathrm{c} \leq \frac{V'(\bar{x}) - V^*}{M(N)} \leq E_\mathrm{r}
	\end{equation}
	for any $N\in\mathbb{N}$, $\bar{x}\in\mathcal{X}_0$, $\theta'\in\mathcal{S}'(\bar{x})$, and $(x_0^*,\theta^*)\in\mathcal{S}^*(\bar{x},\theta')$, where $y_j'=y_j(\bar{x},\theta',U^\mathrm{d})$ and $y^*_j=y_j(x_0^*,\theta^*,U^\mathrm{d})$, $j\in[0,N]$.
	Adding $C_\mathrm{c}$ and applying $\gamma^{-1}\in\mathcal{K}_\infty$ to both sides of \eqref{eq:proof_cumulative} and choosing $\alpha=\alpha_\mathrm{c}$ and $E = \gamma^{-1}(E_\mathrm{r}+C_\mathrm{c})$ establishes \eqref{eq:intTP}, thus completing the proof.
\end{pf}

\begin{lem}[Cumulative turnpike implies dissipativity]\label{lem:intTP2diss}
	If the problem~\eqref{eq:NLP_case_1} has the cumulative turnpike property (Definition~\ref{def:intTP}), then it is strictly dissipative (Definition~\ref{def:diss}).
\end{lem}
\begin{pf}
	Consider $N\in\mathbb{N}$, $\bar{x}\in\mathcal{X}_0$, $\theta'\in\mathcal{S}'(\bar{x})$, and $(x_0^*,\theta^*)\in\mathcal{S}^*(\bar{x},\theta')$. We consider the following candidate storage function:
	\begin{align}
		\lambda^\mathrm{c}_j:=M(N)\sum_{t=j}^N\alpha(\|y'_t-y_t^*\|)-\sum_{t=j}^Ns_t \label{eq:proof_candidate_storage}
	\end{align}
	with $\alpha$ from Definition~\ref{def:intTP} and where $y'_j=y_j(\bar{x},\theta',U^\mathrm{d})$, $y^*_j=y_j(x_0^*,\theta^*,U^\mathrm{d})$, and $s_j = s_j((\bar{x},\theta'),(x_0^*,\theta^*))$, $j\in[0,N]$.
	We first establish the dissipation inequality in~\eqref{eq:diss_dissip}.
	To this end, we evaluate $\lambda^\mathrm{c}_j$ at the successor point $j+1$:
	\begin{align*}
		\lambda^\mathrm{c}_{j+1} &= M(N)\sum_{t=j+1}^N\alpha(\|y'_t-y^*_t\|)-\sum_{t=j+1}^Ns_t\\
		&= M(N)\sum_{t=j}^N\alpha(\|y'_t-y^*_t\|)-\sum_{t=j}^Ns_t\\ &\qquad -(M(N)\alpha(\|y'_j-y^*_j\|)-s_j)\\
		&= \lambda^\mathrm{c}_j -M(N)\alpha(\|y'_j-y^*_j\|)+s_j,
	\end{align*}
	where the last step follows from \eqref{eq:proof_candidate_storage}.
	Hence, $\lambda^\mathrm{c}_j$ satisfies~\eqref{eq:diss_dissip} for all $j\in[0,N]$ with $\alpha_\lambda=\alpha$.
	
	It remains to verify the initial and terminal conditions on the storage function $\lambda^\mathrm{c}_j$. First, \eqref{eq:proof_candidate_storage} directly implies $\lambda^\mathrm{c}_{N+1}=0$. For $j=0$, using \eqref{eq:supply_value} and the definition of the value function $V'(\bar{x})$ from \eqref{eq:NLP_case_1}, \eqref{eq:proof_candidate_storage} evaluates to
	\begin{equation*}
		\lambda^\mathrm{c}_0
		= M(N) \sum_{t=0}^N\alpha(\|y'_t-y^*_t\|)-(V'(\bar{x})-V^*).
	\end{equation*}
	Using the cumulative turnpike property (Definition~\ref{def:intTP}) and the fact that $V^*\leq V'(\bar{x})$ by optimality, we obtain that
	\begin{equation*}
		\lambda^\mathrm{c}_0\leq M(N)\sum_{t=0}^N\alpha(\|y'_t-y^*_t\|) \leq M(N)\cdot E.
	\end{equation*}
	Letting $E_\lambda=E$ verifies the desired uniform upper bound as stated in Definition~\ref{def:diss} and thus completes the proof.
\end{pf}

\begin{lem}[Strict dissipativity implies coercivity]\label{lem:diss2coerc}
	If the SEM problem in~\eqref{eq:NLP_case_1} is strictly dissipative (Definition~\ref{def:diss}), then the value function in~\eqref{eq:NLP_case_1} is coercive (Definition~\ref{def:coerc}).
\end{lem}
\begin{pf}
	From the dissipation inequality in~\eqref{eq:diss_dissip}, we have
	\begin{equation*}
		\begin{split}
			&\lambda_{N+1}((\bar{x},\theta'),(x^*_0,\theta^*)) -  \lambda_0((\bar{x},\theta'),(x^*_0,\theta^*)) \nonumber \\
			&\leq \sum_{j=0}^{N}s_j -  M(N)\sum_{j=0}^{N}\alpha_\lambda(\|y_j'-y_j^*\|)
		\end{split}
	\end{equation*}
	for any $N\in\mathbb{N}$, $\bar{x}\in\mathcal{X}_0$, $\theta'\in\mathcal{S}'(\bar{x})$, and $(x_0^*,\theta^*)\in\mathcal{S}^*(\bar{x},\theta')$, where $y'_j=y_j(\bar{x},\theta',U^\mathrm{d})$, $y^*_j=y_j(x_0^*,\theta^*,U^\mathrm{d})$, and $s_j = s_j((\bar{x},\theta'),(x_0^*,\theta^*))$, $j\in[0,N]$.
	By application of~\eqref{eq:supply_value} and the initial and terminal conditions on the storage function (see Definition~\ref{def:diss}), we can infer that
	\begin{align*}
		V'(\bar{x})
		\geq&\ V^* +  M(N)\sum_{j=0}^N \alpha_\lambda(\|y'_j-y^*_j\|) \\
		&\ +  \lambda_{N+1}(\bar{x},\theta',(x^*_0,\theta^*)) - \lambda_0(\bar{x},\theta',(x^*_0,\theta^*))\\
		\geq&\ V^* +  M(N)\sum_{j=0}^N \alpha_\lambda(\|y'_j-y^*_j\|) - M(N) E_\lambda,
	\end{align*}
	i.e.,~\eqref{eq:coerc} holds with $\gamma(s) = s$, $\alpha_\mathrm{c} = \alpha_\lambda$, and $C_\mathrm{c}=E_\lambda$. Hence, we can conclude that the value function $V'(\bar{x})$ is coercive in the sense of Definition~\ref{def:coerc}, completing the proof.
\end{pf}

% Print the credit authorship contribution details
\printcredits

\section*{Declaration of generative AI use in the manuscript preparation process}
During the preparation of this work, the authors used ChatGPT (OpenAI) to support language editing, manuscript organization, and exploratory discussions of mathematical arguments, examples, and formulations. All mathematical results, proofs, interpretations, and conclusions were independently assessed and verified by the authors. The authors reviewed and edited all AI-assisted content and take full responsibility for the content of the article.

%% Bibliography

% Biography
\bio{images/schiller}
\textbf{Julian D. Schiller} received his Master's degree in Mechatronics in 2019 and his Ph.D. degree in Electrical Engineering in 2025, both from the Leibniz University Hannover, Germany. He is currently a Postdoctoral Researcher at the Institute of Automatic Control, Leibniz University Hannover, Germany. His research interests are in the area of optimization-based state estimation and control of nonlinear systems, turnpike theory, and their interface with deep learning.
\endbio

\bio{images/mueller}
\textbf{Matthias A. Müller} received a Diploma degree in engineering cybernetics from the University of Stuttgart, Germany, an M.Sc. in electrical and computer engineering from the University of Illinois at Urbana-Champaign (both in 2009), and a Ph.D. in mechanical engineering from the University of Stuttgart in 2014. Since 2019, he is Director of the Institute of Automatic Control and Full Professor at the Leibniz University Hannover, Germany.

His research interests include nonlinear control and estimation, model predictive control, and data- and learning-based control, with application in different fields including biomedical engineering and robotics. He has received various distinctions for his work, including the European Systems \& Control PhD Thesis Award, an ERC Starting Grant from the European Research Council, the IEEE CSS George S. Axelby Outstanding Paper Award, the Brockett-Willems Outstanding Paper Award, and the Journal of Process Control Paper Award. He serves/d as an associate editor for Automatica and as an editor of the International Journal of Robust and Nonlinear Control.
\endbio

\end{document}